\documentclass[copyright]{eptcs}
\providecommand{\event}{T. Neary, D. Woods, A.K. Seda and N. Murphy (Eds.):\\ The Complexity of Simple Programs 2008.}
\providecommand{\volume}{1}
\providecommand{\anno}{2009}
\providecommand{\firstpage}{108}
\usepackage{breakurl}
\usepackage{amsmath,amssymb,amsthm}

\newcommand{\To}{\Longrightarrow}
\newcommand{\ve}{\varepsilon}
\newcommand{\tor}{\longrightarrow}

\def\ie{\emph{ i.e.}}

\newtheorem{thm}{Theorem}

\title{Computational Power of P Systems with Small Size Insertion and Deletion Rules}
\author{Alexander Krassovitskiy
\email{alexander.krassovitskiy@estudiants.urv.cat}
\institute{Rovira i Virgili University, \\
Research Group on Mathematical Linguistics,\\
Pl. Imperial T\`arraco 1, 43005 Tarragona, Spain}
\and 
Yurii Rogozhin
\email{rogozhin@math.md}
\institute{Rovira i Virgili University, \\
Research Group on Mathematical Linguistics,\\
Pl. Imperial T\`arraco 1, 43005 Tarragona, Spain}
\institute{Institute of Mathematics and Computer Science\\
Academy of Sciences of Moldova, Academiei, 5, MD-2028, Moldova}
\and
Sergey Verlan
\email{verlan@univ-paris12.fr}
\institute{Institute of Mathematics and Computer Science\\
Academy of Sciences of Moldova, Academiei, 5, MD-2028, Moldova}
\institute{
LACL, D\'epartement Informatique, Universit\'e Paris Est,\\
61 av. G\'en\'eral de Gaulle, 94010 Cr\'eteil, France }
}
\def\titlerunning{Computational Power of P Systems with Small Size Insertion and Deletion Rules}
\def\authorrunning{A. Krassovitskiy,  Y. Rogozhin and  S. Verlan}

\begin{document}
\maketitle

\begin{abstract}
Recent investigations show insertion-deletion systems of small size that are
not complete and cannot generate all recursively enumerable languages. However,
if additional computational distribution mechanisms like P systems are added,
then the computational completeness is achieved in some cases. In this article
we take two insertion-deletion systems that are not computationally complete, consider them in
the framework of P systems and show that the computational power is strictly
increased by proving that any recursively enumerable language can be generated.
At the end some open problems are presented.
\end{abstract}


\section{Introduction}

The operations of insertion and deletion are fundamental in formal language
theory, and generative mechanisms based on them were considered (with
linguistic motivation) for some time, see \cite{Marcus} and \cite{Galiuk}.
Related formal language investigations can be found in several places; we
mention only \cite{Kari}, \cite{KariThierrin}, \cite{MVPS}, \cite{Kluwer}. In
the last years, the study of these operations has received a new motivation
from molecular computing, see \cite{Daley}, \cite{cross}, \cite{DNA},
\cite{TY}, \cite{cfinsdel}.

In general form, an insertion operation means adding a substring to a given
string in a specified (left and right) context, while a deletion operation
means removing a substring of a given string from a specified  (left and right)
context. A finite set of insertion-deletion rules, together with a set of
axioms provide a language
generating device (
an InsDel system): starting from the set of initial strings and iterating
insertion-deletion operations as defined by the given rules we get a language.
The number of axioms, the length of the inserted or deleted strings, as well as
the length of the contexts where these operations take place are natural
descriptional complexity measures in this framework. As expected, insertion and
deletion operations with context dependence are very powerful, leading to
characterizations of recursively enumerable languages. Most of the papers
mentioned above contain such results, in many cases improving the complexity of
insertion-deletion systems previously available in the literature.

Some combinations of parameters lead to systems which are not computationally
complete~\cite{MRV07}, \cite{KRV08} or even decidable~\cite{SV2-2}. However, if
these systems are combined with the distributed computing framework of P
systems~\cite{membr}, then their computational power may strictly increase,
see~\cite{KRV08b} where non-complete insertion-deletion systems of size
(1,1,0;1,1,0) can generate any RE language, if considered in a P systems
framework. In this paper we continue investigation of P systems with
insertion-deletion and we show that P systems with insertion-deletion of size
(2,0,0;1,1,0) and (1,1,0;2,0,0) are computationally complete, while pure
insertion-deletion systems of the same size are not \cite{KRV08c}.

\section{Prerequisites}\label{sec:def}

All formal language notions and notations we use here are elementary and
standard. The reader can consult any of the many monographs in this area -- for
instance, \cite{handbook} -- for the unexplained details.

We denote by $|w|$ the length of a word $w$
and by $card(A)$ the cardinality of the set $A$.

An {\it InsDel system} is a construct $ID=(V,T,A,I,D),$ where $V$ is an
alphabet, $T\subseteq V$, $A$ is a finite language over $V$, and $I,D$ are
finite sets of triples of the form $(u,\alpha,v)$, $\alpha\ne\varepsilon$,
where $u$ and $v$ are strings over $V$ and $\ve$ denotes the empty string. The elements of
$T$ are {\it terminal} symbols (in contrast, those of $V-T$ are called
nonterminals), those of $A$ are {\it axioms}, the triples in $I$ are {\it
insertion rules}, and those from $D$ are {\it deletion rules}. An insertion
rule $(u,\alpha,v)\in I$ indicates that the string $\alpha$ can be inserted in
between $u$ and $v$, while a deletion rule $(u,\alpha,v)\in D$ indicates that
$\alpha$ can be removed from the context $(u,v)$. As stated otherwise,
$(u,\alpha,v)\in I$ corresponds to the rewriting rule $uv\to u\alpha v$, and
$(u,\alpha,v)\in D$ corresponds to the rewriting rule $u\alpha v\to uv$. We
denote by $\To_{ins}$ the relation defined by an insertion rule (formally,
$x\To_{ins} y$ iff $x=x_1uvx_2,y=x_1u\alpha vx_2$, for some $(u,\alpha,v)\in I$
and $x_1,x_2\in V^*$) and by $\To_{del}$ the relation defined by a deletion
rule (formally, $x\To_{del} y$ iff $x=x_1u\alpha vx_2,y=x_1u vx_2$, for some
$(u,\alpha,v)\in D$ and $x_1,x_2\in V^*$). We refer by $\To$ to any of the
relations $\To_{ins},\To_{del}$, and denote by $\To^*$
the reflexive and transitive closure of $\To$ (as usual, $\To^+$ is its 
transitive closure).

The language generated by $ID$ is defined by $$L(ID)=\{w\in T^*\mid x\To^* w,
 x\in A\}.$$

The complexity of an InsDel system $ID=(V,T,A,I,D)$ is traditionally described
by the vector $(n,m;p,q)$ called \emph{weight}, where
\vspace{-2mm}
\begin{eqnarray*}
&&n=\max\{|\alpha|\mid (u,\alpha,v)\in I\},\\
&&m=\max\{|u|\mid (u,\alpha,v)\in I \ {\rm or}\ (v,\alpha,u)\in I\},\\
&&p=\max\{|\alpha|\mid (u,\alpha,v)\in D\},\\
&&q=\max\{|u|\mid (u,\alpha,v)\in D \ {\rm or}\ (v,\alpha,u)\in D\},
\end{eqnarray*}

The {\it total weight} of $ID$ is the sum $\gamma=m+n+p+q$.
\smallskip

However, it was shown in~\cite{SV2-2} that this complexity measure is not
accurate and it cannot distinguish between universality and non-universality
cases (there are families having the same total weight but not the same
computational power). In the same article it was proposed to use the length of
each context instead of the maximum. More exactly,
\vspace{-2mm}
\begin{eqnarray*}
&&n=\max\{|\alpha|\mid (u,\alpha,v)\in I\},\\
&&m=\max\{|u|\mid (u,\alpha,v)\in I\},\\
&&m'=\max\{|v|\mid (u,\alpha,v)\in I\},\\
&&p=\max\{|\alpha|\mid (u,\alpha,v)\in D\},\\
&&q=\max\{|u|\mid (u,\alpha,v)\in D\},\\
&&q'=\max\{|v|\mid (u,\alpha,v)\in D\}.
\end{eqnarray*}

\noindent Hence the complexity of an insertion-deletion system will be
described by the vector
 $(n,m,m';p,q,q')$ that we call \emph{size}. We also
denote by $INS^{m,m'}_nDEL^{q,q'}_p$ corresponding families of
insertion-deletion systems. Moreover, we define the total weight of the system
as the sum of all numbers above: $\psi=n+m+m'+p+q+q'$. Since it is known
from~\cite{SV2-2} that systems using a context-free insertion or deletion of
one symbol are not powerful, we additionally require $n+m+m'\ge 2$ and
$p+q+q'\ge 2$.

If some of the parameters $n,m,m',p,q,q'$ is not specified, then we write
instead the symbol~$*$. In particular, $INS^{0,0}_*DEL^{0,0}_*$ denotes the
family of languages generated by {\em context-free InsDel systems}. If one of
numbers from the couples $m$, $m'$ and/or $q$, $q'$ is equal to zero (while the
other is not), then we say that corresponding families have a one-sided
context.


InsDel systems of a ``sufficiently large'' weight can characterize $RE$, the
family of recursively enumerable languages.


An \emph{insertion-deletion P system} is the following construct:
\begin{align*}
\Pi=(V,T,\mu,M_1,\dots{},M_n,R_1,\dots,R_n),	
\end{align*}
where

\begin{itemize}
\item $V$ is a finite alphabet,
\item $T\subseteq V$ is the terminal alphabet,
\item $\mu$ is the membrane (tree) structure of the system which has $n$
    membranes (nodes). This structure will be represented by a word
    containing correctly nested marked parentheses.
\item $M_i,$ for each $1 \le i \le n$ is a finite language associated to
    the membrane $i$.
\item $R_i,$ for each $1 \le i \le n$ is a set of insertion and deletion
    rules with target indicators associated to membrane $i$ and which have
    the following form: $(u,x,v;tar)_a$, where $(u,x,v)$ is an insertion rule,
    and $(u,x,v;tar)_e$, where $(u,x,v)$ is an deletion rule,  and $tar$,
    called the \emph{target indicator}, is from the set $\{here,in,out\}$.
\end{itemize}


Any m-tuple $(N_1, \dots{},  N_n)$ of languages over V is called a
configuration of $\Pi$. For two configurations $(N_1, \dots{},
N_n)$ and $(N_1',\dots{}, N_n')$ of $\Pi$  we write $(N_1,
\dots{}, $ $ N_n) \To (N_1', \dots{} , N_n')$ if we can pass from
$(N_1, \dots{},  N_n)$ to $(N_1', \dots{}, N_m')$ by applying the
insertion and deletion rules from each region of $\mu$, in maximally parallel way,
\ie, in parallel
to all possible strings from the corresponding regions,
and following the target indications associated with the rules.
We assume that every string represented in membrane has arbitrary many copies.
Hence, by applying a rule to a string we get both arbitrary many copies
of resulted string as well as old copies of the same string.

More specifically, if $w \in M_i$ and $r = (u,x,v;tar)_a\in R_i$,
respectively  $r = (u,x,v;tar)_e\in R_i$, such that $w \To^r_{ins}
w'$, respectively $w \To^r_{del} w'$,  then $w'$ will go to the
region indicated by $tar$. If $tar = here$, then the string
remains in $M_i$, if $tar = out$, then the string is moved to the
region immediately outside the membrane $i$ (maybe, in this way
the string leaves the system), if $tar = in$, then the string is
moved to the region immediately below.


A sequence of transitions between configurations of a given insertion-deletion
P system $\Pi$, starting from the initial configuration $(M_1, \dots, M_n)$, is
called a computation with respect to $\Pi$. The result of a computation
consists of all strings over $T$ which are sent out of the system at any time
during the computation. We denote by $L(\Pi)$ the language of all strings of
this type. We say that $L(\Pi)$ is generated by $\Pi$.

We denote by $ELSP_k(insdel, (n,m,m';p,q,q'))$(see, for example \cite{membr})
the family of languages $L(\Pi)$ generated by insertion-deletion P systems of
degree at most $k, k \ge 1$ having the size $(n,m,m';p,q,q')$.

\section{Main results}

\begin{thm}\label{thm:MEMBR_RE110200}
 $ELSP_5(insdel, (1,1,0; 2,0,0)) = RE$.
\end{thm}

\begin{proof}
We prove the inclusion $$ELSP_5(insdel, (1,1,0; 2,0,0)) \supseteq RE$$ by
simulating a type-0 grammar in Penttonen normal form by the means of
insertion-deletion systems. The reverse inclusion $$ELSP_5(insdel, (1,1,0;
2,0,0))\subseteq RE$$ is obvious as it follows from the Church thesis.
\medskip

Let $G = (N, T, S, R)$ be a type-0 grammar in Penttonen normal form.
This means that all production rules in $R$ are of the form:
\begin{align*}
AB \tor AC &\mbox{~or~}
\\
A \tor BC &\mbox{~or~}
\\
A \tor \alpha
\end{align*}
\noindent where $A, B$ and $C$ are from $N$ and $\alpha \in
T\cup N \cup\{ \ve \}.$ Suppose that rules in $R$ are ordered and
$n=card(R)$.

\noindent Now consider the following system.

\noindent  $\Pi_1 = (V, T, [_1~[_2~[_3~[_4~[_5~]_5~]_4~]_3~]_2~]_1, \{SX\},
\emptyset, \emptyset, \emptyset, \emptyset,  R_1, R_2, R_3, R_4, R_5)$.

\noindent It has a new nonterminal alphabet $V = N\cup T \cup \overline{P}
\cup \{ X \}, \overline{P} = \{P^j_i |i=1, \ldots ,n, \ j=1,
\ldots, 4 \}.$

\begin{itemize}
\item

For every production $i: AB \tor AC$ from $R$ with
$A, B, C \in N$ we add following rules to $R_1, \ldots, R_4$
correspondingly (we do not use membrane 5 in this case):
\begin{align*}
&(A, P^1_i, \ve; in)_a \mbox{~to~} R_1;
\\
&(P^1_i, P^2_i, \ve ;in)_a \mbox{~and~} (\ve, P^1_iP^3_i,\ve
;out)_e \mbox{~to~} R_2;
\\
&(\ve, P^2_iB, \ve ;in)_e \mbox{~and~} (P^3_i, C,\ve ;out)_a
\mbox{~to~} R_3;
\\
&(P^1_i, P^3_i, \ve ;out)_a \mbox{~to~} R_4;
\end{align*}

\item
 For every production $i: A \tor BC$ from $R$ where $A,
B, C \in N$ we add rules:
\begin{align*}
&(A, P^1_i, \ve; in)_a \mbox{~to~} R_1;
\\
&(P^1_i, P^2_i, \ve ;in)_a \mbox{~and~} (\ve, P^2_i,\ve ;out)_e \mbox{~to~} R_2;
\\
&(P^1_i, B, \ve ;in)_a \mbox{~and~} (\ve, P^3_i, \ve ;out)_e
\mbox{~to~} R_3;
\\
&(\ve, AP^1_i, \ve ;in)_e \mbox{~and~} ( P^3_i, C, \ve ;out)_a \mbox{~to~} R_4;
\\
&(P^2_i, P^3_i, \ve; out)_a \mbox{~to~} R_5.
\end{align*}

\item
 For every production $i: A \tor \alpha$ from $R$ where
$A\in N, \alpha \in T\cup N$ we add following rules to $R_1, \ldots, R_4$
correspondingly (we do not use membrane 5 in this case):
\begin{align*}
&(A, P^1_i, \ve; in)_a \mbox{~to~} R_1;
\\
&( P^1_i,\alpha, \ve ;in)_a \mbox{~and~} (\ve, P^2_iP^3_i, \ve ;out)_e
\mbox{~to~} R_2; \\
&(P^1_i, P^2_i, \ve ;in)_a \mbox{~and~} (P^2_i,P^3_i,\ve ;out)_a \mbox{~to~} R_3;
\\
&(\ve, AP^1_i, \ve ;out)_e \mbox{~to~} R_4;
\end{align*}

\item
For every production $i: A \tor \ve$ from $R$ with $A\in N$ we add
rules $(\ve, A, \ve; here)_e$ to $R_1$.
\smallskip

\item
Finally, we add to $R_1$ rule $(\ve,X,\ve; out)_e$.
\end{itemize}

\medskip

We claim that $\Pi_1$ generates the same language as $G$. In fact it
is enough to proof that every step in derivation by grammar $G$
can be simulated in $\Pi_1$.
\medskip

Let us consider production
$i: AB \tor AC \in R$.
\smallskip

The simulation of this rule is controlled by symbols $P^1_i$, $P^2_i$ and
$P^3_i$. We assume that the sentential form in the skin membrane does not
contain symbols from $\overline{P}$. Consider a string $w_1 AB w_2$ in the skin
region. We insert $P^1_i$ after symbol $A: \ w_1 AB w_2 \To w_1  AP^1_iB w_2$
and send the obtained string to membrane 2. Here we insert $P^2_i$ after symbol
$P^1_i: \ w_1  AP^1_iB w_2 \To w_1  AP^1_iP^2_iB w_2$ and send the string to
membrane 3. Next we delete substring $P^2_iB: \ w_1A P^1_iP^2_iB w_2 \To w_1
AP^1_i w_2$ and send the obtained string to membrane 4. Here we insert $P^3_i$
after $P^1_i: \ w_1 AP^1_i w_2 \To w_1 AP^1_i P^3_i w_2$ and push  the string
to membrane 3. Now we insert symbol $C$ after  $P^3_i: w_1 AP^1_i P^3_i w_2 \To
w_1 AP^1_iP^3_iC w_2$ pushing the string to membrane 2. Now we have two
possibilities: to delete substring $P^1_iP^3_i$ and push the result $w_1ACw_2$
to the skin membrane (thus we simulate rule $i: AB \tor AC \in R$ correctly),
or to insert symbol $P^2_i$ after $P^1_i$ and send  string
$w_1AP^1_iP^2_iP^3_iCw_2$ to membrane 3, where symbol $C$ will be inserted and
the string comes back to membrane 2. So, we have a circle of computation in
membrane 2 and 3. Notice, that between symbols $P^1_i$ and $P^3_i$ there is at
least one symbol $P^2_i$, and therefore there is no possibility to apply rule
$(\ve, P^1_iP^3_i, \ve; out)_e$ and to enter at the skin membrane. So, this
branch of computation cannot influence the result and may be omitted in the
consideration.
\medskip

Let us consider production  $i : A \tor BC$, where $A,B,C \in N$.
\smallskip

The simulation of this rule is controlled by symbols $P^1_i$, $P^2_i$ and
$P^3_i$. We can also assume that the sentential form in the skin membrane does
not contain symbols from $\overline{P}$. Consider a string $w_1 AB w_2$ in the
skin region. We insert $P^1_i$ after symbol $A: \ w_1 AB w_2 \To w_1 AP^1_iB
w_2$ and send the obtained string to membrane 2. Here we insert $P^2_i$ after
symbol $P^1_i: \ w_1  AP^1_iB w_2 \To w_1  AP^1_iP^2_iB w_2$ and send the
string to membrane 3. Here we insert symbol $B$ after $P^1_i: \ w_1A P^1_iP^2_i
w_2 \To w_1 AP^1_iBP^2_i w_2$ and send the obtained string to membrane 4. Here
we delete substring $AP^1_i :  \ w_1 AP^1_iBP^2_iw_2 \To w_1 B P^2_i w_2$ and
send the string to membrane 5. Now we insert symbol $P^3_i$ after  $P^2_i: w_1
BP^2_i w_2 \To w_1 BP^2_iP^3_i w_2$ and push the string to membrane 4. Here we
insert symbol $C$ after symbol $P^3_i : w_1BP^2_iP^3_iw_2 \To
w_1BP^2_iP^3_iCw_2$ and push the string to membrane 3. Here we delete symbol
$P^3_i$ and push the string to membrane 2: $w_1BP^2_iP^3_iCw_2 \To
w_1BP^2_iCw_2$. At last we delete symbol $P^2_i$ and the result $w_1BCw_2$
enters at the skin region. So, we simulate rule $i :  \ A \tor BC$ correctly.
\medskip

Simulation of production  $i : \ A \tor \alpha$, where $A \in N$ and $\alpha \in N \cup T$
is done in an analogous manner.
\smallskip

Every $\ve$-production $i: A \tor \ve$, $A \in N$ is simulated directly in the skin
membrane by the corresponding rule $(\ve,A,\ve; here)_e$.

According to the definition of insertion-deletion P systems the
result of a computation consists of all strings over $T$ which are
sent out of the system at any time during the computation. This is
formally provided by the rule $(\ve,X,\ve; out)_e$ in the skin
membrane. This rule uses conventional notation from \cite{membr}.
Indeed, assume a sentential form $wX$ appears in the skin
membrane for some $w \in T^*$ (as we stared from the axiom $SX$).
Then, applying the rule $(\ve,X,\ve; out)_e$ we assure that $w$ is in $L(\Pi_1)$.
\medskip

To claim the proof we observe that every correct sentential form
has at most one symbol $P^1_i$, $P^2_i$ or $P^3_i$, $i=1,\ldots,n.$ And after
insertion of $P^1_i$ in the skin membrane either all rules
corresponding to $i$-th rule have to be applied (in the defined
order) or the derivation is blocked. Hence, we have $L(G) =
L(\Pi_1).$
\end{proof}

\begin{thm}\label{thm:MEMBR_RE200110}
 $ELSP_5(insdel, (2,0,0; 1,1,0)) = RE$.
\end{thm}

\begin{proof}
We prove the inclusion
$$ELSP_5(insdel, (1,1,0; 2,0,0)) \supseteq RE$$
by simulating a type-0 grammar in Penttonen normal form. The reverse inclusion
$$ELSP_5(insdel, (1,1,0; 2,0,0))\subseteq RE$$ follows from the Church thesis.
\smallskip

Let  $G = (N, T, S, R)$ be a type-0 grammar in Penttonen normal form with
production rules $R$ are of type:
\begin{align*}
AB \tor AC &\mbox{~or~}
\\
A \tor BC &\mbox{~or~}
\\
A \tor \alpha
\end{align*}
\noindent where $A, B, C$ and $D$ are from $N$ and $\alpha \in T\cup N \cup\{
\ve \}.$ Suppose that rules in $R$ are ordered and $n=card(R)$.

\noindent Now consider the following system.

\noindent  $\Pi_2 = (V, T, [_1~[_2~[_3~[_4~[_5~]_5~]_4~]_3~]_2~]_1,
\{SX\}, \emptyset, \emptyset, \emptyset, \emptyset, R_1, R_2, R_3,
R_4, R_5)$.

\noindent It has a new nonterminal alphabet $V = N\cup T \cup
\overline{P} \cup \{ X \}$, $\overline{P} = \{P^j_i |i=1, \ldots ,n, \ j=1,
\ldots, 5 \}.$

\begin{itemize}
\item
For every production $i: AB \tor AC$ from $R$ with $A, B, C \in N$
we add following rules to $R_1, \ldots, R_4$ correspondingly:
\begin{align*}
&(\ve, P^1_i P^2_i, \ve; in)_a \mbox{~to~} R_1;
\\
&(P^2_i, B,\ve ;in)_e \mbox{~and~} (A, P^3_i,\ve ;out)_e
\mbox{~to~} R_2;
\\
&(\ve, P^3_i C, \ve ;in)_a \mbox{~and~} (A, P^2_i,\ve ;out)_e
\mbox{~to~} R_3;
\\
&(A, P^1_i, \ve ;out)_e;
\end{align*}

\item
For every production $i: A \tor BC$ from $R$ with $A, B, C \in N$
we add following rules to $R_1, \ldots, R_5$ correspondingly:
\begin{align*}
&(\ve, P^1_i P^2_i, \ve; in)_a \mbox{~to~} R_1;
\\
&(P^2_i, A,\ve ;in)_e \mbox{~and~} (\ve, P^3_i,\ve ;out)_e
\mbox{~to~} R_2;
\\
&(\ve, B P^3_i, \ve ;in)_a \mbox{~and~}(P^3_i, P^2_i,\ve ;out)_e
\mbox{~to~} R_3;
\\
&(P^3_i, P^1_i, \ve ;in)_e \mbox{~and~} (P^2_i, P^4_i,\ve ;out)_e
\mbox{~to~} R_4;
\\
&(\ve, P^4_iC, \ve ;out)_a \mbox{~to~} R_5;
\end{align*}

\item
For every production $i: A \tor \alpha$ from $R$ with $A \in N, \alpha \in N\cup T$
we add following rules to $R_1, \ldots, R_4$:
\begin{align*}
&(\ve, \alpha P^3_i, \ve; in)_a \mbox{~to~} R_1;
\\
&(P^3_i, A,\ve ;in)_e \mbox{~and~} (\alpha, P^2_i,\ve ;out)_e
\mbox{~to~} R_2;
\\
&(\ve, P^1_i P^2_i, \ve ;in)_a \mbox{~and~} (\alpha, P^1_i,\ve ;out)_e
\mbox{~to~} R_3;
\\
&(\alpha, P^3_i, \ve ;out)_e;
\end{align*}

\item
For every production $i: A \tor \ve$ from $R$ with $A \in N$ we
add the following rule to $R_1$: $(\ve, A, \ve; here)_e.$
\smallskip

\item
Finally, we add to $R_1$ the rule $(\ve,X,\ve; out)_e$.
\end{itemize}

Now we claim that $\Pi_2$ generates the same language as $G$. We show
that every step in derivation by grammar $G$
can be simulated in $\Pi_2$.
\smallskip

Let us consider  production
$i: AB \tor AC \in R$.
\smallskip

The simulation of this rule is controlled
by symbols $P^1_i$, $P^2_i$ and $P^3_i$.
As in the previous theorem, we assume that sentential
form in the first membrane does not contain symbols from
$\overline{P}$.
Insertion of two symbols $P^1_i P^2_i$ sends the
sentential form to the second membrane. As at this moment there are no
symbols $P^3_i$ the only possible rule to be applied is $(P^2_i,
B,\ve ;in)_e$. It assumes the presence of $B$ on the right of $P^2_i$.
This rule sends the sentential form to the third membrane. At this
moment we can only apply the insertion $(\ve, P^3_i C, \ve ;in)_a$
which sends the form to the forth membrane (hence $(A, P^2_i,\ve
;out)_e$ requires symbol $A$ on the right from $P^2_i$). In the
forth membrane we can apply the deletion rule $(A, P^1_i, \ve ;out)_e$
only if the first insertion $P^1_i P^2_i$ was done between $A$ and
$B$. Now we are pushed back to the third membrane. Here we have
two options. The first option is to repeat the insertion $(\ve,
P^3_i C, \ve ;in)_a$. The derivation will be blocked in the next
step as there is no symbols $P^1_i$ anymore. The second option is to apply
$(A, P^2_i,\ve ;out)_e$. This
is always possible since symbol $P^2_i$ appears adjacently right
from $A$. This sends the sentential form to the second membrane.
At this moment the sentential form does not contain any symbols
from $\overline{P}$ except for $P^3_i$. And we can apply the
deletion rule $(A, P^3_i,\ve ;out)_e$ assuming $P^3_i C$ is
inserted adjacently right from $P^1_i P^2_i$.

Hence, the only possible derivation by using the rules above is the
following:
\begin{align*}
& w_1ABw_2\To w_1AP^1_iP^2_iBw_2 \To w_1AP^1_iP^2_iw_2 \To \\
& w_1AP^1_iP^2_i P^3_i Cw_2 \To w_1A P^2_i P^3_i Cw_2 \To \\
& w_1AP^3_i Cw_2 \To w_1A Cw_2.
\end{align*}
One can see that this derivation correctly simulates the rule $i:
AB \tor AC$.
\smallskip

Now we consider a context-free rule $i: A \tor BC$, where $A,B,C
\in N$.
\smallskip

The simulation of the rule is controlled
by symbols $P^1_i$, $P^2_i$, $P^3_i$ and $P^4_i$.
The rule $(\ve, P^1_i P^2_i, \ve; in)_a$ inserts $P^1_i P^2_i$ and
sends the sentential form to the second membrane. In the second
membrane deletion rule $(P^2_i, A,\ve ;in)_e$ is applicable if
$P^1_i P^2_i$ is inserted adjacently left from $A$. It
sends the form to the third membrane. Here, only insertion rule
$(\ve, B P^3_i, \ve ;in)_a$ is applicable as at this moment there
are no symbols $P^3_i$ yet. It sends the form to the forth
membrane. Here we can only delete  $P^1_i$ as the rule
$(P^2_i,P^4_i, \ve ;out)_e$ cannot be applied.
In the fifth membrane we insert $P^4_iC$ and the sentential form
is pushed back to the forth membrane. At this step we can only
remove $P^4_i$ and send the string to membrane 3. Now we have two
possibilities: either insertion rule $(\ve, BP^3_i, \ve; in)_a$ or
deletion rule $(P^3_i,P^2_i,\ve; out)_e$ can be applied. In the
first case the derivation will be blocked in membrane 4, as no
rules may be applied to the string. In the second case  symbol
$P^2_i$ will be deleted and the string enters at membrane 2. Here
symbol $P^3_i$ will be deleted and the result $w_1BCw_2$ appears
at the skin membrane.

Hence, the only possible derivation by using these rules is the
following:
\begin{align*}
& w_1Aw_2\To w_1P^1_iP^2_iAw_2 \To w_1P^1_iP^2_iw_2 \To \\
& w_1 BP^3_iP^1_iP^2_i w_2 \To w_1 BP^3_iP^2_iw_2 \To \\
& w_1 BP^3_i P^2_iP^4_iCw_2 \To w_1 BP^3_iP^2_iC w_2 \To \\
& w_1 B P^3_iC w_2 \To w_1 B C w_2.
\end{align*}

So, we simulate rule $i: A \tor BC$ correctly.

Now, consider production $i: A \tor \alpha$ from $R$ with $A \in N, \alpha \in N\cup T$.
This case of replacement basically uses one insertion of $\alpha P^3_i$
adjacently left from $A$, and two deletion rules $(P^3_i, A,\ve
;in)_e$ and $(\alpha, P^3_i, \ve ;out)_e$. But, hence, the total number
of insertion-deletion rules for every production has to be even,
we introduce one additional insertion $(\ve, P^1_i P^2_i, \ve
;in)_a$ and two deletion rules $(\alpha, P^1_i,\ve ;out)_e$, and  $(\alpha,
P^2_i,\ve ;out)_e$.

The derivation for this case has the following form:
\begin{align*}
& w_1Aw_2\To w_1\alpha P^3_iAw_2 \To w_1\alpha P^3_iw_2 \To w_1\alpha P^3_i
P^1_iP^2_i w_2 \To \\
&w_1 \alpha P^1_iP^2_i w_2 \To w_1 \alpha P^2_i w_2 \To w_1 \alpha w_2
\end{align*}

So, we simulate rule $i: A \tor \alpha$ correctly.

Every $\ve$-production $i: A \tor \ve$, $A \in N$ is simulated directly in the skin
membrane by the corresponding rule $(\ve,A,\ve; here)_e$. Finally, the rule
$(\ve,X,\ve; out)_e$ is applied
to $wX$ in the skin membrane, where $w \in T^*$ and $X$ is from the axiom $SX$.
Here, we use the same technique as in the previous theorem.
This rule is needed in order to terminate derivation and sent the
resulting string as an output of the system.

In order to finish the proof we observe that every correct sentential form
preserves the following properties:
\begin{enumerate}
\item No symbol from $\overline{P}$ presents in the skin membrane.
\item If some symbol
from $\overline{P}$ appears more than once in the sentential form
than the derivation is blocked on this production.
\end{enumerate}

As shown before
insertion of $P^1_i P^2_i$ or $B P^3_i$ for the corresponding
$i-th$ rule in the skin membrane results to either all rules
corresponding to $i$-th rule have to be applied (in the defined
order) or the derivation is blocked. Hence, we have $L(G) =
L(\Pi_2).$
\end{proof}


\section{Conclusions}\label{sec:conclusions}

In this article we have investigated P systems based on small size
insertion-deletion systems. We proved two universality results, namely that
insertion-deletion P systems with 5 membranes of size (2,0,0;1,1,0) and
(1,1,0;2,0,0) are computationally complete. At the same time, pure
insertion-deletion systems of the same size are not computationally complete. We guess that
their computational power is rather small, but its precise characterizations is
an open question. Another interesting question is whether the number of
membranes used in the proof of Theorems~\ref{thm:MEMBR_RE200110}
and~\ref{thm:MEMBR_RE110200} is minimal.

Finally, we would like to mention an interesting decidable class of
insertion-deletion systems: systems of size $(2,0,0; 2,0,0)$. We think that P
systems with rules from this class will still not be able to generate any
recursively enumerable language.

\section*{Acknowledgments} The first author acknowledges the
grant of Ramon y Cajal from University Rovira i Virgili 2005/08
and grant no. MTM 63422 from the Ministry of Science and Education
of Spain. The second author acknowledges the support of European
Commission, project MolCIP, MIF1-CT-2006-021666. The second and
the third author acknowledge the Science and Technology Center in
Ukraine, project 4032.

\bibliographystyle{eptcs}

\begin{thebibliography}{}
\providecommand{\bibitemstart}[1]{\bibitem{#1}}
\providecommand{\bibitemend}{}
\providecommand{\bibliographystart}{}
\providecommand{\bibliographyend}{}
\providecommand{\url}[1]{\texttt{#1}}
\providecommand{\urlprefix}{Available at }
\providecommand{\bibinfo}[2]{#2}
\bibliographystart

\bibliographyend
\end{thebibliography}


\begin{thebibliography}{13}

\bibitem{Daley}
M. Daley, L. Kari, G. Gloor, R. Siromoney, Cir\-cu\-lar con\-tex\-tual
in\-ser\-ti\-ons/de\-le\-ti\-ons  with appli\-ca\-ti\-ons to bio\-mole\-cu\-lar
com\-pu\-ta\-tion. In: {\em Proc. of 6th Int. Symp. on String
Processing and Information Retrieval, SPIRE'99\/} (Cancun, Mexico, 1999), 
47--54.

\bibitem{Galiuk} B.S. Galiukschov, Semicontextual grammars,
{\em Matematika Logica i Matematika Linguistika\/}, Tallin University, 1981
38--50 (in Russian).



\bibitem{Kari} L. Kari, {\em On insertion and deletion in formal
languages\/}, PhD Thesis, University of Turku, 1991.

\bibitem{cross} L. Kari, Gh. P\u aun, G. Thierrin, S. Yu, At the crossroads of
    DNA com\-pu\-ting and formal languages: characterizing RE using
    insertion-deletion systems. In: {\em  Proc. of 3rd DIMACS Workshop
    on DNA Based Computing\/}, Philadelphia, 1997, 318--333.

\bibitem{KariThierrin} L. Kari, G. Thierrin, Contextual
insertion/deletion and computability, {\em Information and Computation\/}, {\bf
131}, 1 (1996), 47--61.

\bibitem{KRV08}A. Krassovitskiy, Yu. Rogozhin, S. Verlan, Further results on
    insertion-deletion systems with one-sided contexts, Pre-proceedings of the 2nd International
Conference on Language and Automata. Theory and Application. LATA 2008, March
13-19, 2008. Techical Reports of Research Group on Mathematical Linguistics,
No. 36/08, 2008, 347--358.

\bibitem{KRV08b}A. Krassovitskiy, Yu. Rogozhin, S. Verlan, One-sided Insertion and Deletion: Traditional and P Systems Case,
Proceedings of International Workshop on Computing with Biomolecules, August 27th, 2008, Wien, Austria, 53--64.

\bibitem{KRV08c}A. Krassovitskiy, Yu. Rogozhin, S. Verlan, About computational completeness of one-sided insertion-deletion systems (submitted).


\bibitem{Marcus} S. Marcus, Contextual grammars, {\em Rev. Roum. Math.
Pures Appl.\/}, {\bf 14} (1969), 1525--1534.


\bibitem{cfinsdel} M. Margenstern, Gh. P\u aun, Yu. Rogozhin, S. Verlan, Context-free
insertion-deletion systems. \emph{Theoretical Computer Science}, \textbf{330}
(2005), 339--348.


\bibitem{MVPS} C. Martin-Vide, Gh. P\u aun, A. Salomaa,
Characterizations of recursively enumerable languages by means of insertion
grammars, {\em Theoretical Computer Science\/}, {\bf 205}, 1--2 (1998),
195--205.


\bibitem{MRV07}
A.Matveevici, Yu.Rogozhin, S.Verlan, Insertion-Deletion Systems with One-Sided Contexts.
{\em Lecture Notes in Computer Science\/}, Springer, vol. 4664 (2007) 205-217.

\bibitem{Kluwer} Gh. P\u aun, {\em Marcus contextual grammars\/}. Kluwer, Dordrecht, 1997.

\bibitem{membr} Gh. P\u aun, {\em Membrane Computing. An Introduction\/}
(Springer--Verlag, Berlin, 2002), 163, 226--230.

\bibitem{DNA} Gh. P\u aun, G. Rozenberg, A. Salomaa, {\em DNA Computing. New
Computing Paradigms\/}. Springer--Verlag, Berlin, 1998.

\bibitem{handbook} G. Rozenberg, A. Salomaa, eds., {\em Handbook of Formal Languages\/}.
Springer--Verlag, Berlin, 1997.


\bibitem{TY} A. Takahara, T. Yokomori, On the computational power of
    insertion-deletion systems. In: {\em Proc. of 8th International
    Workshop on DNA-Based Computers, DNA8\/} (Sapporo, Japan, June 10--13,
    2002), {\em Revised Papers, LNCS}, {\bf 2568} (2003), 269--280.



\bibitem{SV2-2} S. Verlan, On Minimal Context-Free Insertion-Deletion Systems.
 \emph{Journal of Automata, Languages and Combinatorics} 12 (2007) 1/2, 317–-328.

\end{thebibliography}

\end{document}